\documentclass[11pt,reqno]{article}
\usepackage{latexsym, amsmath, amssymb, a4, epsfig}
\usepackage{amsthm}
\usepackage{graphicx}
\usepackage{caption}
\usepackage{subcaption}
\usepackage{latexsym, color}
\usepackage{booktabs}

\newtheorem{thm}{Theorem}[section]

\newtheorem{lem}[thm]{Lemma}

\usepackage[left=1in,right=1in,top=1in,bottom=1in]{geometry}

\newcommand{\eqnref}[1]{(\ref {#1})}

\newcommand{\beq}{\begin{equation}}
\newcommand{\eeq}{\end{equation}}

\newcommand{\bef}{\begin{figure}}
\newcommand{\enf}{\end{figure}}

\title{Asymptotic Analysis of the Narrow Escape Problem in Dendritic Spine Shaped Domain: Three Dimension}

\author{Hyundae Lee\thanks{Department of Mathematics, Inha University, Incheon
402-751, Korea (hdlee@inha.ac.kr).} \and Xiaofei Li\thanks{\footnotesize Department of Mathematics, South University of Science and Technology of China, Shenzhen, China (xiaofeilee@hotmail.com). Corresponding author.} \and Yuliang Wang\thanks{\footnotesize Department of Mathematics, Hong Kong Baptist University, Kowloon Tong, Hong Kong SAR (jadelightking@qq.com).}}

\begin{document}

\maketitle

\begin{abstract}
This paper deals with the three-dimensional narrow escape problem in dendritic spine shaped domain, which is composed of a relatively big head and a thin neck. The narrow escape problem is to compute the mean first passage time of Brownian particles traveling from inside the head to the end of the neck. The original model is to solve a mixed Dirichlet-Neumann boundary value problem for the Poisson equation in the composite domain, and is computationally challenging. In this paper we seek to transfer the original problem to a mixed Robin-Neumann boundary value problem by dropping the thin neck part, and rigorously derive the asymptotic expansion of the mean first passage time with high order terms. This study is a nontrivial generalization of the work in \cite{Li}, where a two-dimensional analogue domain is considered.
\end{abstract}
\section{Introduction}
The narrow escape problem (NEP) in diffusion theory, which goes back to Lord Rayleigh \cite{lord}, is to calculate the mean first passage time (MFPT) of a Brownian particle to a small absorbing window on the otherwise reflecting boundary of a bounded domain. NEP has recently attracted significant attention from the point of view of mathematical and numerical modeling due to its relevance in molecular biology and biophysics. The small absorbing window often represents a small target on a cellular membrane, such as a protein channel, which is a target for ions \cite{hille}, a receptor for neurotransmitter molecules in a neuronal synapse \cite{elias}, a narrow neck in the neuronal spine, which is a target for calcium ions \cite{kork}, and so on.  A main concern for NEP is to derive an asymptotic expansion of the MFPT when the size of the small absorbing window  tends to zero. There have been several significant works deriving the leading-order and higher-order
terms of the asymptotic expansions of MFPT for regular and singular domains in two and three dimensions \cite{agkls,AKL,benichou,harris,cluster,review,laws,simon,leakage}.

In this paper we focus on the MFPT of calcium ion in the three-dimensional dendritic spine shape domain. For many neurons in the mammalian brain, the postsynaptic terminal of an excitatory synapse is found in a specialized structure protruding from the dendritic draft, known as a dendritic spine (Figure~\ref{spine}(a)). Dendritic spines function as biochemical compartments that regulate the duration and spread of postsynaptic calcium fluxes produced by glutamatergic neurotransmission at the synapses \cite{thomasG}. Calcium ions are ubiquitous signaling molecules that accumulate in the cytoplasm in response to diverse classes of stimuli and, in turn, regulate many aspects of cell function \cite{michaelJ}. The rate of diffusional escape from spines through their narrow neck, which we call MFPT, is one factor that regulates the retention time of calcium ions in dendritic spines.

Most spines have a bulbous head, and a thin neck that connects the head of the spine to the shaft of the dendrite (Figure~\ref{spine}). Spines are separated from their parent dendrites by this thin neck and compartmentalize calcium during synaptic stimulation. It has been given that many organelles inside the spine head do not affect the nature of the random motion of ions, mainly due to their large size relative to that of ions \cite{modeling}. So we assume the environment inside the spine is homogenous, which means the movement of the calcium ions is pure diffusion.

\bef
\centering
\begin{subfigure}[b]{0.4\textwidth}
  \includegraphics[width=\textwidth]{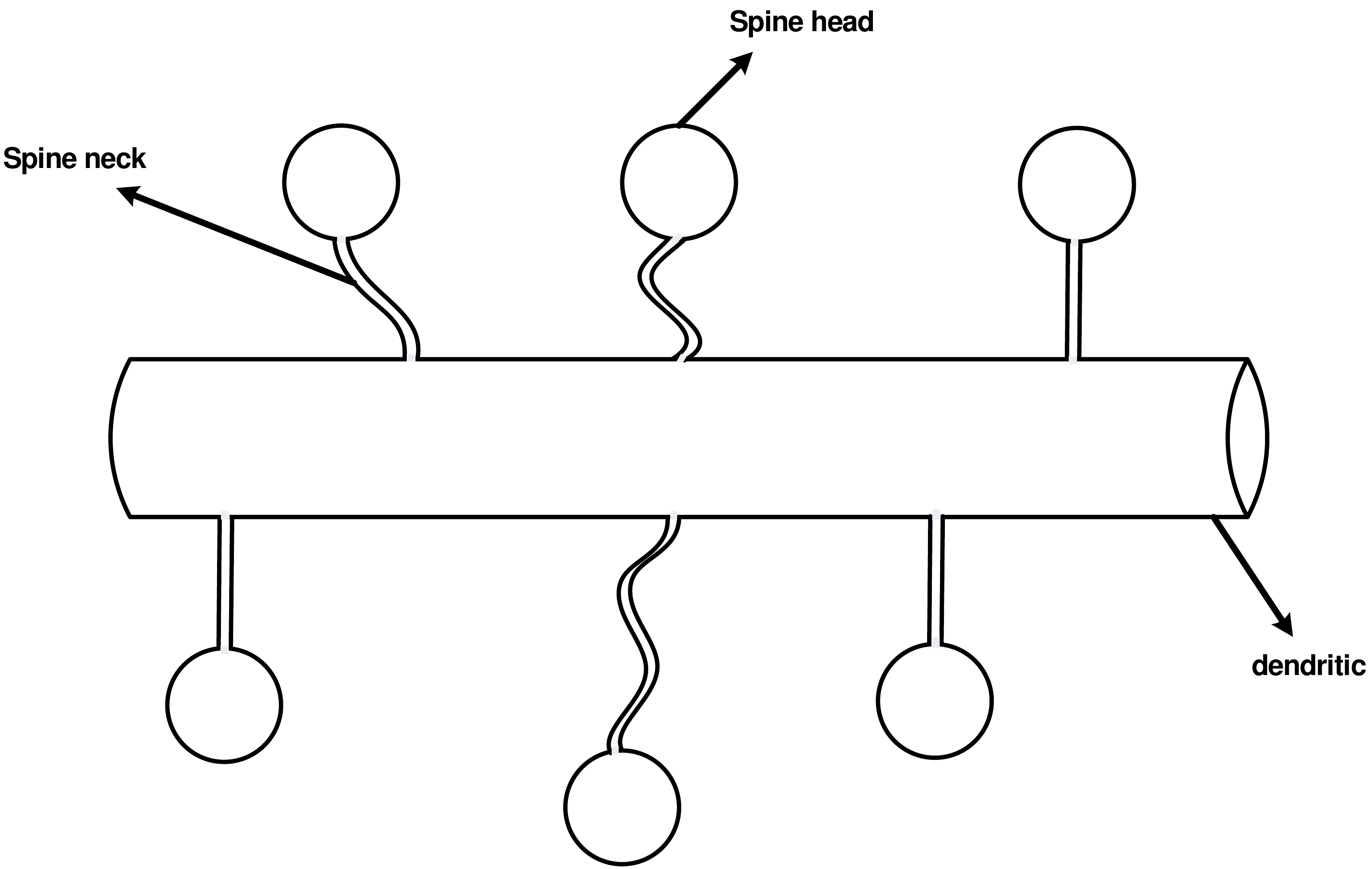}\hspace{1cm}
  \caption{}
  \end{subfigure}
  \hspace{1cm}
  \begin{subfigure}[b]{0.4\textwidth}
  \includegraphics[width=\textwidth]{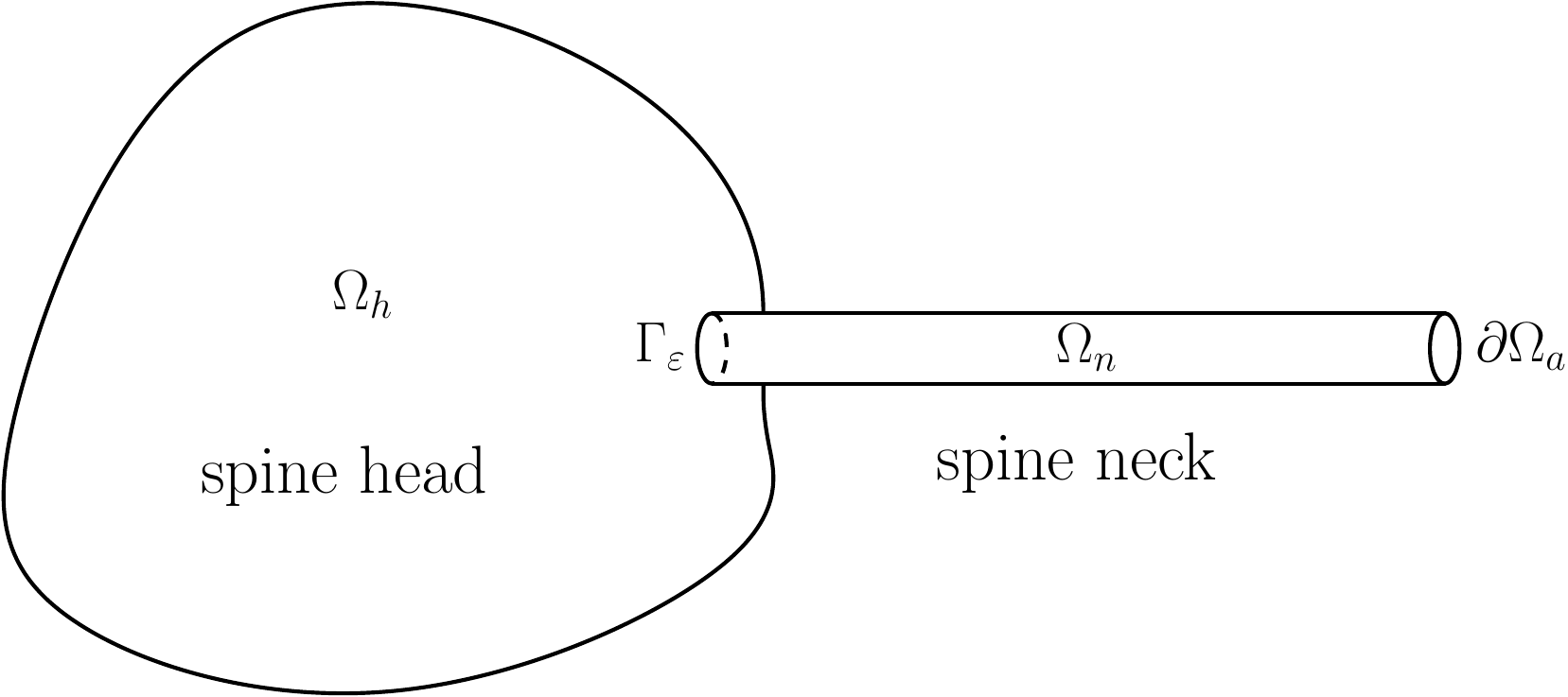}
  \caption{}
  \end{subfigure}
\caption{(a).Dendritic spine morphology. (b).The modeling shape of dendritic spine $\Omega=\Omega_h\cup \Omega_n$ with a spherical spine head $\Omega_h$ and a cylindrical spine neck $\Omega_n$. }
  \label{spine}
 \enf

The NEP can be mathematically formulated in the following way. Let $\Omega_h$ be a bounded simply connected domain in $\mathbb{R}^{3}$. Let $\Omega_n$ be a cylinder with length $L$ and radius $r=O(\varepsilon)$ which is much smaller than the length. Connecting this cylinder with $\Omega_h$, we have the geometry for the spine (Figure~\ref{spine}(b)). The connection part between $\Omega_h$ and $\Omega_n$ is a small interface which is denoted by $\Gamma_{\varepsilon}$.  Let $\Omega=\Omega_h\cup \Omega_n$ denote the domain for the whole spine. Suppose that the boundary $\partial \Omega$ is decomposed into the reflecting part $\partial \Omega_{r}$ and the absorbing part $\partial \Omega_{a}$, where $\partial \Omega_{a}$ is the end of the thin cylinder neck. We assume that the area of $\partial \Omega_{a}$, $|\partial \Omega_{a}|=O(\varepsilon^2)$ is much smaller than the area of the whole boundary. The NEP is to calculate the MFPT $u$ which is the unique solution to the following boundary value problem,
\beq \label{main_equation}
   \begin{cases}
         \Delta u=-1,\quad &\mbox{in}~\Omega,\\[1ex]
         \dfrac{\partial u}{\partial \nu}=0,&\mbox{on}~\partial \Omega_{r},\\[1ex]
         u=0,&\mbox{on}~\partial \Omega_{a},
         \end{cases}
         \eeq
where $\nu$ is the outer unit normal to $\partial\Omega$. The asymptotic analysis for NEP arises in deriving the asymptotic expansion of $u$ as
$\varepsilon \rightarrow 0$, from which one can estimate the escape time of the calcium ions. In \cite{ania, michael}, the first order of the asymptotic expansion has been obtained numerically as
\begin{align}\label{eq:1}
  u\approx \frac{|\Omega_h|L}{\pi\varepsilon^2},
\end{align}
where $|\Omega_h|$ denotes the volume of the spine head. In this study, we derive higher order asymptotic solution to (\ref{main_equation}) by means of the Neumann-Robin model which is proposed in \cite{Li} to deal with the narrow escape time in a two-dimensional analogue domain.

In the Robin-Neumann model the solution to the original boundary value problem \eqref{main_equation} in the singular domain $\Omega$ is approximated by the solution to the following boundary value problem in the smooth domain $\Omega_h$:
\beq \label{NR_eq}
\begin{cases}
\Delta u_{\varepsilon}=-1,\quad &\mbox{in}~\Omega_h,\\[1ex]
\dfrac{\partial u_{\varepsilon}}{\partial \nu}=0,&\mbox{on}~\partial \Omega_{r},\\[1em]
\dfrac{\partial u_{\varepsilon}}{\partial \nu}+\alpha u_{\varepsilon}=\beta,&\mbox{on}~\Gamma_\varepsilon:=\partial \Omega_h \setminus \overline{\partial \Omega_{r}},
\end{cases}
\eeq
where $\Omega_h$ is the spine head of $\Omega$ mentioned in Figure~\ref{spine}, and $\Gamma_\varepsilon:=\partial \Omega_h \setminus \overline{\partial \Omega_{r}}$ is the connection part between the spine head and the spine neck. Here $\alpha>0$ and $\beta>0$ are constants to be determined. We assume $\alpha < \alpha_0$ for some a priori constant $\alpha_0>0$ and $\varepsilon$ is sufficiently small so that $\alpha \varepsilon \ll 1$. We shall apply the layer potential technique to derive the asymptotic solution $u_{\varepsilon}$ to (\ref{NR_eq}) as
\begin{align*}
u_{\varepsilon}(x) \approx \frac{|\Omega_h|}{\pi\alpha\varepsilon^2}
+\frac{|\Omega_h| M}{\pi^2\varepsilon}+\frac{\beta}{\alpha}-\frac{|\Omega_h|}{2\pi|x-x^*|}
\end{align*}
for $x \in \Omega_h$ and away from $\Gamma_\varepsilon$, where $M$ is a computable constant, $x^*$ is a fixed point in $\Gamma_\varepsilon$.

This study is organized as follows. In section 2, we review the Neumann function for the Laplacian in $\mathbb{R}^3$, which is a major tool for our study. In section 3, we derive the asymptotic solution for Robin-Neumann model. In section 4, we apply the Robin-Neumann boundary model to approximate the MFPT of calcium ion in dendritic spine. Numerical experiments are also given in this section to confirm the theoretical results. This study ends with a short conclusion in section 5.

\section{Neumann function in $\mathbb{R}^3$}

Let $\Omega$ be a bounded domain in $\mathbb{R}^3$ with $C^2$ smooth boundary $\partial\Omega$, and let $N(x,z)$ be the Neumann function for $-\Delta$ in $\Omega$ with a given $z\in\Omega$. That is, $N(x,z)$ is the solution to the boundary value problem
\begin{equation*}
  \begin{cases}
        \Delta_{x} N(x,z)=-\delta_{z}, &x\in \Omega,\\[1ex]
      \displaystyle\frac{\partial N}{\partial \nu_{x}}=-\frac{1}{|\partial \Omega|},  & x\in \partial\Omega,\\[1em]
      \displaystyle\int_{\partial \Omega}N(x,z)d\sigma(x)=0,&
  \end{cases}
\end{equation*}
where $\nu$ is the outer unit normal to the boundary $\partial\Omega$.

If $z\in \Omega$, then $N(x,z)$ can be written in the form
\begin{eqnarray*}
 N(x,z)=\frac{1}{4\pi|x-z|}+R_{\Omega}(x,z),\quad x\in\Omega,
 \end{eqnarray*}
where $R_{\Omega}(x,z)$ has weaker singularity than $1/|x-z|$ and solves the boundary value problem

\begin{equation*}
  \begin{cases}
-\Delta_{x} R_{\partial\Omega}(x,z)=0,&x\in\Omega,\\[1ex]
\displaystyle\frac{\partial R_{\Omega}}{\partial \nu_{x}}\Big|_{x\in \partial \Omega}=-\frac{1}{|\partial \Omega|}+\frac{1}{4\pi}\frac{\langle x-z,\nu_{x}\rangle}{|x-z|^{3}},&x\in \partial \Omega.
\end{cases}
\end{equation*}
where $\langle\cdot,\cdot\rangle$ denotes the inner product in $\mathbb{R}^3$.

If $z\in \partial \Omega$, then Neumann function on the boundary is denoted by $N_{\partial \Omega}$ and can be written as
\begin{align}\label{eq:4}
N_{\partial \Omega}(x,z)=\frac{1}{2\pi|x-z|}+R_{\partial\Omega}(x,z),\quad x \in \Omega,z \in \partial \Omega,
\end{align}
where $R_{\partial \Omega}(x,z)$ has weaker singularity than $1/|x-z|$ and solves the boundary value problem

\begin{equation*}
  \begin{cases}
        \Delta_{x}R_{\partial \Omega}(x,z)=0,& x\in \Omega,\\[1ex]
       \displaystyle\frac{\partial R_{\partial\Omega}}{\partial \nu_{x}}\Big|_{x\in \partial \Omega}=-\frac{1}{|\partial \Omega|}+\frac{1}{2\pi}\frac{\langle x-z,\nu_{x}\rangle}{|x-z|^{3}},&x\in \partial \Omega, ~z\in \partial \Omega.
       \end{cases}
\end{equation*}
The structure of $R_{\partial\Omega}$ is given in \cite{popov} as
\beq
R_{\partial\Omega}(x,z)=-\frac{1}{4\pi}H(z)\ln|x-z|+v_{\partial\Omega}(x,z),
\label{NLN}
\eeq
where $z\in\partial\Omega$, $x\in\Omega\cup\partial\Omega$, where $H(z)$ denotes the mean curvature of $\partial\Omega$ at $z$, and $v_{\partial\Omega}$ is a bounded
function.

\section{Derivation of the asymptotic expansion}
The goal in this section is to derive the asymptotic expansion of $u_{\varepsilon}$ to \eqref{NR_eq} as $\varepsilon\rightarrow 0$. For simplicity we assume the connection part $\Gamma_\varepsilon$ lies in a plane. The general case where $\Gamma_\varepsilon$ is curved can be handled with minor modifications.

\begin{thm} \label{thm:main}
The solution $u_{\varepsilon}$ to the boundary value problem \eqnref{NR_eq} has the following asymptotic expansion,
\begin{equation*}
u_{\varepsilon}(x)=\frac{|\Omega_h|}{\pi\alpha\varepsilon^2}
+\frac{|\Omega_h| M}{\pi^2\varepsilon}+\frac{\beta}{\alpha}-\frac{|\Omega_h|}{2\pi|x-x^*|}+\Phi(x)+O(\alpha),
\label{u}
\end{equation*}
where $M$ is a constant given by
\begin{align}
  \label{eq:2}
  M = \int_{\Gamma_1} \int_{\Gamma_1} \frac{1}{2\pi|x-z|} dxdz,
\end{align}
$x^*$ is a fixed point in $\Gamma_\varepsilon$, and $\Phi(x)$ is a bounded function depending only on $\Omega_h$. The remainder $O(\alpha)$ is uniform in $x\in \Omega_h$ satisfying dist$(x,\Gamma_{\varepsilon})\geq c$ for some constant $c>0$.
\end{thm}

Proof. By integrating the first equation in \eqnref{NR_eq} over $\Omega_h$ using the divergence theorem we get the compatibility condition
\beq
\int_{\Gamma_{\varepsilon}}\frac{\partial u_{\varepsilon}}{\partial \nu}d\sigma=-|\Omega_h|.
\label{compatibility}
\end{equation}

Let us define $g(x)$ by
$$g(x)=\int_{\Omega_h}N(x,z)dz,\quad x\in\Omega_h,$$\\
which is seen to solve the boundary value problem
\beq \label{g_eq}
\begin{cases}
\Delta g=-1,& \mbox{in}~\Omega_h,\\[1ex]
\displaystyle\frac{\partial g}{\partial\nu}=-\frac{|\Omega_h|}{|\partial \Omega_h|}, & \mbox{on}~ \partial \Omega_h,\\[1em]
\displaystyle \int_{\partial \Omega_h} g d\sigma=0.&
\end{cases}
\end{equation}
Applying the Green's formula and using \eqnref{NR_eq} and \eqnref{g_eq}, we obtain
\beq \label{u_form1}
u_{\varepsilon}(x)=g(x) + \int_{\Gamma_{\varepsilon}}N_{\partial\Omega_h}(x,z)\frac{\partial u_{\varepsilon}(z)}{\partial \nu_{z}}d\sigma(z)
+C_{\varepsilon},
\eeq
where $$C_{\varepsilon}=\frac{1}{|\partial \Omega_h|}\int_{\partial \Omega_h}u_{\varepsilon}(z)d\sigma(z).$$

Let $x\in\Gamma_{\varepsilon}$. Substitute the Robin boundary condition and the structure of Neumann function (\ref{NLN}) into (\ref{u_form1}), we obtain
     \begin{equation*}
       \frac{\beta}{\alpha}-\frac{1}{\alpha}\phi_{\varepsilon}(x)=g(x)+\frac{1}{2\pi}\int_{\Gamma_{\varepsilon}}\frac{1}{|x-z|}\phi_{\varepsilon}(z)d\sigma(z)
    + \int_{\Gamma_{\varepsilon}}v_{\partial\Omega_h}(x,z)\phi_{\varepsilon}(z)d\sigma(z)+C_{\varepsilon}, \quad x \in \Gamma_\varepsilon,
      \end{equation*}
where $\phi_{\varepsilon}(x)=\partial u_{\varepsilon}(x)/\partial \nu_{x}$. Note that the mean curvature $H(z)=0$ for $z \in \Gamma_\varepsilon$ since $\Gamma_\varepsilon$ is assumed to be flat.

By a simple change of variables, the above equation can be written as
     \beq
     \begin{split}
       \frac{\beta}{\alpha}-\frac{1}{\alpha\varepsilon}\tilde{\phi}_{\varepsilon}(x)=g(\varepsilon x)+\frac{1}{2\pi}\int_{\Gamma_1}\frac{1}{|x-z|}\tilde{\phi}_{\varepsilon}(z)d\sigma(z)+\varepsilon \int_{\Gamma_{1}}v_{\partial\Omega_h}(\varepsilon x,\varepsilon z)\tilde{\phi}_{\varepsilon}(z)d\sigma(z)+C_{\varepsilon},
      \end{split}
      \label{LL}
      \eeq
      where $\Gamma_1=\{x/\varepsilon: x\in\Gamma_{\varepsilon}\}$, and $\tilde{\phi}_{\varepsilon}(x)=\varepsilon\phi_{\varepsilon}(\varepsilon x)$, $x\in\Gamma_1$.

Define two integral operators $L,~ L_1: L^\infty(\Gamma_1)\rightarrow L^\infty(\Gamma_1)$ by
\begin{eqnarray*}
&L[\varphi](x)&=\frac{1}{2\pi}\int_{\Gamma_1}\frac{1}{|x-z|}\varphi(z)dz,\\
&L_1[\varphi](x)&=\int_{\Gamma_1}v_{\partial\Omega_h}(\varepsilon x,\varepsilon z)\varphi( z)dz.
\end{eqnarray*}
Since $v_{\partial\Omega}$ is bounded, one can easily see that $L_{1}$ is bounded independently of $\varepsilon$. The integral operator $L$ is a also bounded (see the proof in Appendix A).

So we can write \eqnref{LL} as
$$\frac{\beta}{\alpha}-\frac{1}{\alpha\varepsilon}\tilde{\phi}_{\varepsilon}(x)=g(\varepsilon x)+(L+\varepsilon L_1)\tilde{\phi}_{\varepsilon}(x)+C_{\varepsilon}.$$
Collecting $\tilde{\phi}_{\varepsilon}$ terms, we have
\beq
\big[I+\alpha\varepsilon(L+\varepsilon L_1)\big]\tilde{\phi}_{\varepsilon}(x)=\varepsilon\alpha\left(\frac{\beta}{\alpha}-g(\varepsilon x)-C_{\varepsilon}\right).
\label{tphi}
\end{equation}
Assume here $\alpha< \alpha_0$ and $\alpha \varepsilon \ll 1$. It is easy to see that
$$\left(I+\alpha\varepsilon (L+ \varepsilon L_1)\right)^{-1}=I-\alpha\varepsilon (L+ \varepsilon L_1)+O(\alpha^2\varepsilon^2).$$
Noting that $g$ is $C^1$ in $\bar{\Omega}_h$, and $g(\varepsilon x)=g(x^*)(1+O(\varepsilon))$ on $\Gamma_1$, we have from (\ref{tphi}) that
\begin{equation*}
\tilde{\phi}_{\varepsilon}(x)=\varepsilon\alpha\left[I-\alpha\varepsilon (L+ \varepsilon L_1)+O(\alpha^2\varepsilon^2)\right](\tilde{C}_{\varepsilon}+O(\varepsilon)).
\end{equation*}
where
\beq
\tilde{C}_{\varepsilon}:=\frac{\beta}{\alpha}-g(x^*)-C_{\varepsilon}.
\label{cc}
\eeq
By the compatibility condition (\ref{compatibility}), we can see that $\tilde{C}_\varepsilon=O((\alpha\varepsilon^2)^{-1})$. Then collecting terms we have
\begin{equation}\label{}
\tilde{\phi}_{\varepsilon}(x)=\varepsilon\alpha\tilde{C}_{\varepsilon}-
(\varepsilon\alpha)^2\tilde{C}_{\varepsilon}(L[1]+\varepsilon L_1[1])
+O(\alpha\varepsilon^2).
\label{tphi1}
\end{equation}

Plug (\ref{tphi1}) into the compatibility condition (\ref{compatibility}), we obtain
\begin{equation*}
\pi\alpha\varepsilon^2\tilde{C}_\varepsilon-\alpha^2\varepsilon^3\tilde{C}_\varepsilon\int_{\Gamma_1}(L[1]+\varepsilon L_1[1])(x)=-|\Omega_h|+O(\alpha\varepsilon^3),
\end{equation*}
which implies
\begin{equation*}
\begin{split}
\tilde{C}_\varepsilon&=\left(I+\frac{\alpha\varepsilon}{\pi}\int_{\Gamma_1}(L[1]+\varepsilon L_1[1])(x)+O(\alpha^2\varepsilon^2)\right)\left(-\frac{|\Omega_h|}{\pi\alpha\varepsilon^2}+O(\varepsilon)\right)\\
&=-\frac{|\Omega_h|}{\pi\alpha\varepsilon^2}-\frac{|\Omega_h|}{\pi^2\varepsilon}M
-|\Omega_h|v_{\partial\Omega}(x^*,x^*)+O(\alpha),
\end{split}
\end{equation*}
where $M=\int_{\Gamma_1}L[1]dx$, and $\int_{\Gamma_1}L_1[1](x)=\pi^2v_{\partial\Omega}(x^*,x^*)+O(\varepsilon)$.
Hence from (\ref{cc}), we have
\beq
C_\varepsilon=\frac{|\Omega_h|}{\pi\alpha\varepsilon^2}+\frac{|\Omega_h|}{\pi^2\varepsilon}M
+|\Omega_h|v_{\partial\Omega}(x^*,x^*)+\frac{\beta}{\alpha}-g(x^*)+O(\alpha).
\label{c}
\eeq

Substitute $\tilde{C_\varepsilon}$ into \eqnref{tphi1}, we have
\begin{equation*}
\tilde{\phi}_{\varepsilon}(x)=-\frac{|\Omega_h|}{\pi\varepsilon}-\frac{\alpha|\Omega_h|}{\pi^2}(M-\pi L[1])+O(\alpha^2\varepsilon),
\end{equation*}
and hence
\begin{align}\label{eq:3}
\frac{\partial u_{\varepsilon}(x)}{\partial \nu_{x}}=\phi_{\varepsilon}(x)=\frac{1}{\varepsilon}\tilde{\phi}_{\varepsilon} \left( \frac{x}{\varepsilon} \right)
=-\frac{|\Omega_h|}{\pi\varepsilon^2}-\frac{\alpha|\Omega_h|}{\pi^2\varepsilon} \left[ M-\pi L[1] \left( \frac{x}{\varepsilon} \right) \right] +O(\alpha^2).
\end{align}

In order to obtain the solution to (\ref{NR_eq}), it remains to calculate the second term in (\ref{u_form1}). Combining \eqref{eq:4}, \eqref{NLN} and \eqref{eq:3} yields
\begin{align}
&\int_{\Gamma_{\varepsilon}}N_{\partial \Omega_h}(x,z)\frac{\partial u_{\varepsilon}(z)}{\partial \nu_{z}}d\sigma(z)\nonumber\\[1ex]
&=\int_{\Gamma_{\varepsilon}}N_{\partial \Omega_h}(x,z)\left[ -\frac{|\Omega_h|}{\pi\varepsilon^2}
-\frac{\alpha|\Omega_h|}{\pi^2\varepsilon} \left( M-\pi L[1] \left( \frac{x}{\varepsilon} \right)+O(\alpha^2)\right) \right] \nonumber d\sigma(z)\\[1ex]
&=-|\Omega_h|N_{\partial \Omega_h}(x, x^*)+O(\alpha \varepsilon) \notag \\[1ex]
&= - \frac{|\Omega_h|}{2\pi|x-x^*|} + |\Omega_h| v_{\partial\Omega_h} (x,x^*).
\label{phi_exp}
\end{align}
provided that ${\rm dist}(x, \Gamma_{\varepsilon}) \geq c$ for some constant $c>0$.

Finally, combining(\ref{u_form1}), (\ref{c}) and (\ref{phi_exp}), we obtain
\begin{equation*}
u_{\varepsilon}(x)=\frac{|\Omega_h|}{\pi\alpha\varepsilon^2}
+\frac{|\Omega_h| M}{\pi^2\varepsilon}+\frac{\beta}{\alpha}-\frac{|\Omega_h|}{2\pi|x-x^*|}+\Phi(x)+O(\alpha),
\end{equation*}
where
\begin{align*}
\Phi(x)=g(x)-g(x^*)+|\Omega_h|\left[ v_{\partial\Omega_h}(x^*,x^*)-v_{\partial\Omega_h}(x,x^*) \right]
\label{phii}
\end{align*}
is a bounded function depending only on $\Omega_h$.

\section{Application to the narrow escape problem in dendritic spine}
In this section we use the asymptotic solution to the Robin-Neumann model to approximate the calcium ion diffusion time in a dendritic spine domain. Without loss of generality, let $\Omega_n$ be placed along the $x_1$-axis with $\Gamma_\varepsilon$ at $x_1=0$ and $\partial\Omega_a$ at $x_1=L$. Since $\varepsilon \ll 1$, we assume $u(x)$ is constant in each cross section of $\Omega_n$ and thus is a function of $x_1$ only. The three-dimensional problem \eqref{main_equation} restricted in $\Omega_n$ is then approximated by the one-dimensional problem
\begin{equation*}
  \begin{cases}
    \dfrac{d^2 u}{dx_1^2}= -1, & 0<x_1<L, \\[1ex]
    u = 0, & x_1=L.
  \end{cases}
\end{equation*}
By direct calculation we obtain the solution to the above problem as
\begin{equation*}
  u(x_1) = -\dfrac{1}{2} x_1^2 + C x_1 + \dfrac{1}{2} L^2 - CL, \quad 0<x_1<L,
\end{equation*}
where $C$ is a constant. Solution $u$ satisfies the Robin boundary condition at $x=0$. Evaluating $u$ and $du/dx$ at $x_1=0$ yields the Robin condition
\begin{equation*}
  \frac{du}{dx_1}(0) + \frac{1}{L} u(0) = \frac{L}{2}.
\end{equation*}
By the continuity of $u$ and $\partial u / \partial \nu$ on $\Gamma_\varepsilon$, we obtain the Robin-Neumann boundary value problem in $\Omega_h$:
\begin{equation*}
  \begin{cases}
    \Delta u_{\varepsilon}=-1, &\mbox{in}~\Omega_{h},\\[1ex]
    \dfrac{\partial u_{\varepsilon}}{\partial \nu}=0, &\mbox{on}~\partial \Omega_{r},\\[1em]
    \dfrac{\partial u_{\varepsilon}}{\partial \nu}+\alpha u_\varepsilon = \beta, &\mbox{on}~\Gamma_{\varepsilon}.
  \end{cases}
  \label{cal}
\end{equation*}
where $\alpha=1/L$ and $\beta=L/2$. Applying Theorem \ref{thm:main} we obtain the asymptotic expansion of $u_\varepsilon$ as
\begin{align}  \label{eq:5}
  u_{\varepsilon}(x) \approx \frac{|\Omega_h|L}{\pi\varepsilon^2} +\frac{|\Omega_h| M}{\pi^2\varepsilon}+\frac{L^2}{2}-\frac{|\Omega_h|}{2\pi|x-x^*|},
\end{align}
Note that the leading term coincides with \eqref{eq:1}, which is obtained in \cite{ania, michael} using numerical simulation.

\subsection{Numerical experiments}
In the rest of this section we shall conduct numerical experiments to verify the asymptotic expansion \eqref{eq:5}. We shall compare the asymptotic solution \eqref{eq:5} with the solution $u$ to the original problem \eqref{main_equation} obtained numerically with the finite element method. We shall also confirm the coefficients in the first two terms in \eqref{eq:5}.

For simplicity we confine ourself to the case when the connection part $\Gamma_\varepsilon$ is a disk of radius $\varepsilon$. In this case the constant $M$ in $\eqref{eq:2}$ has an explicit and elegant value $M=8/3$ (see Appendix B). The spine neck $\Omega_n$ is chosen to be a cylinder with radius $\varepsilon$ and length $L$, and whose axis is parallel to the normal of $\Gamma_\varepsilon$. To nondimensionalize our problem, the numerical results are regarded as using consistent units throughout this section.
\begin{figure}
\centering
\begin{subfigure}[b]{0.31\textwidth}
    \includegraphics[width=\textwidth]{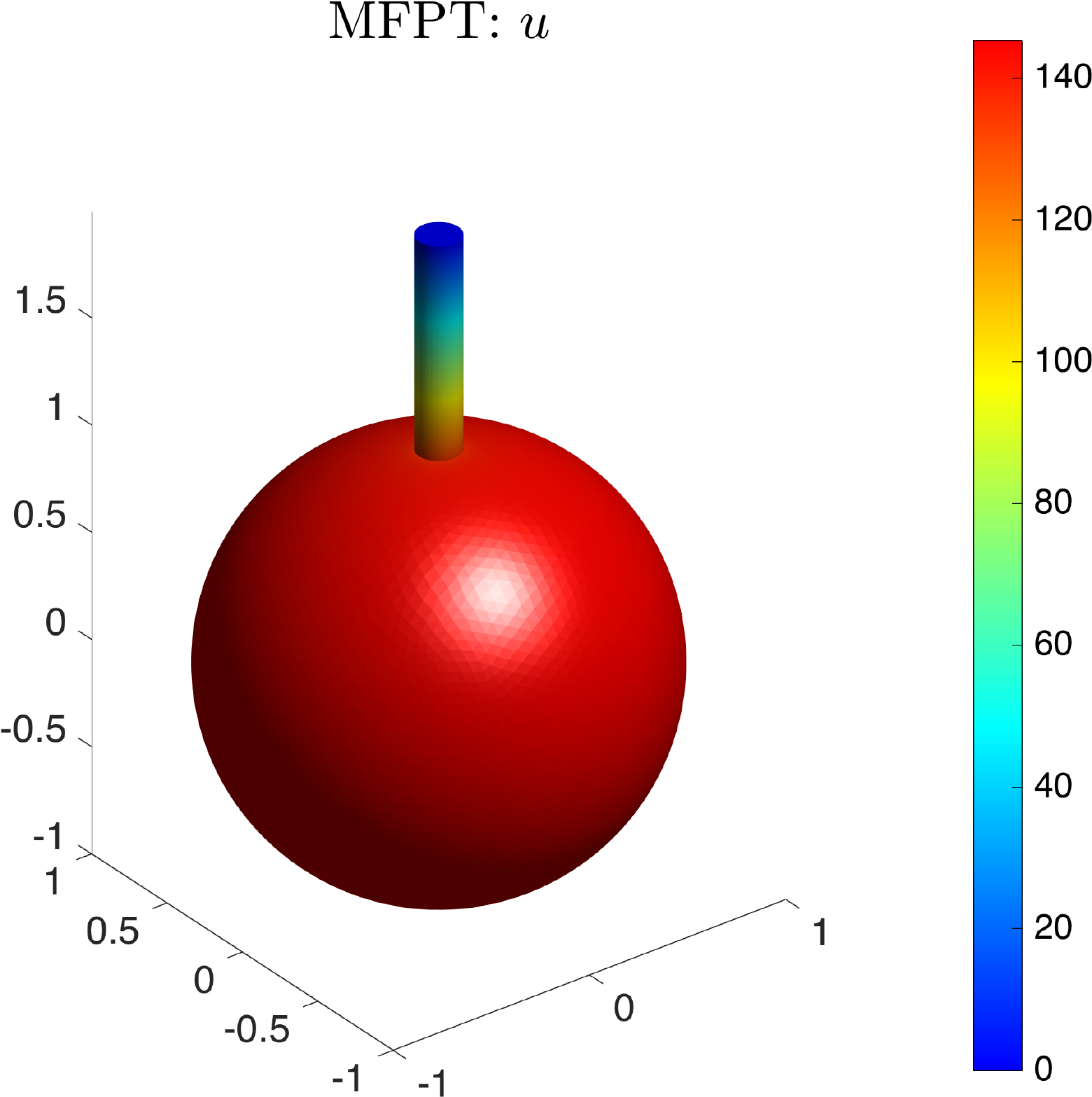}
    \caption{}
    \end{subfigure}
    \begin{subfigure}[b]{0.3\textwidth}
    \includegraphics[width=\textwidth]{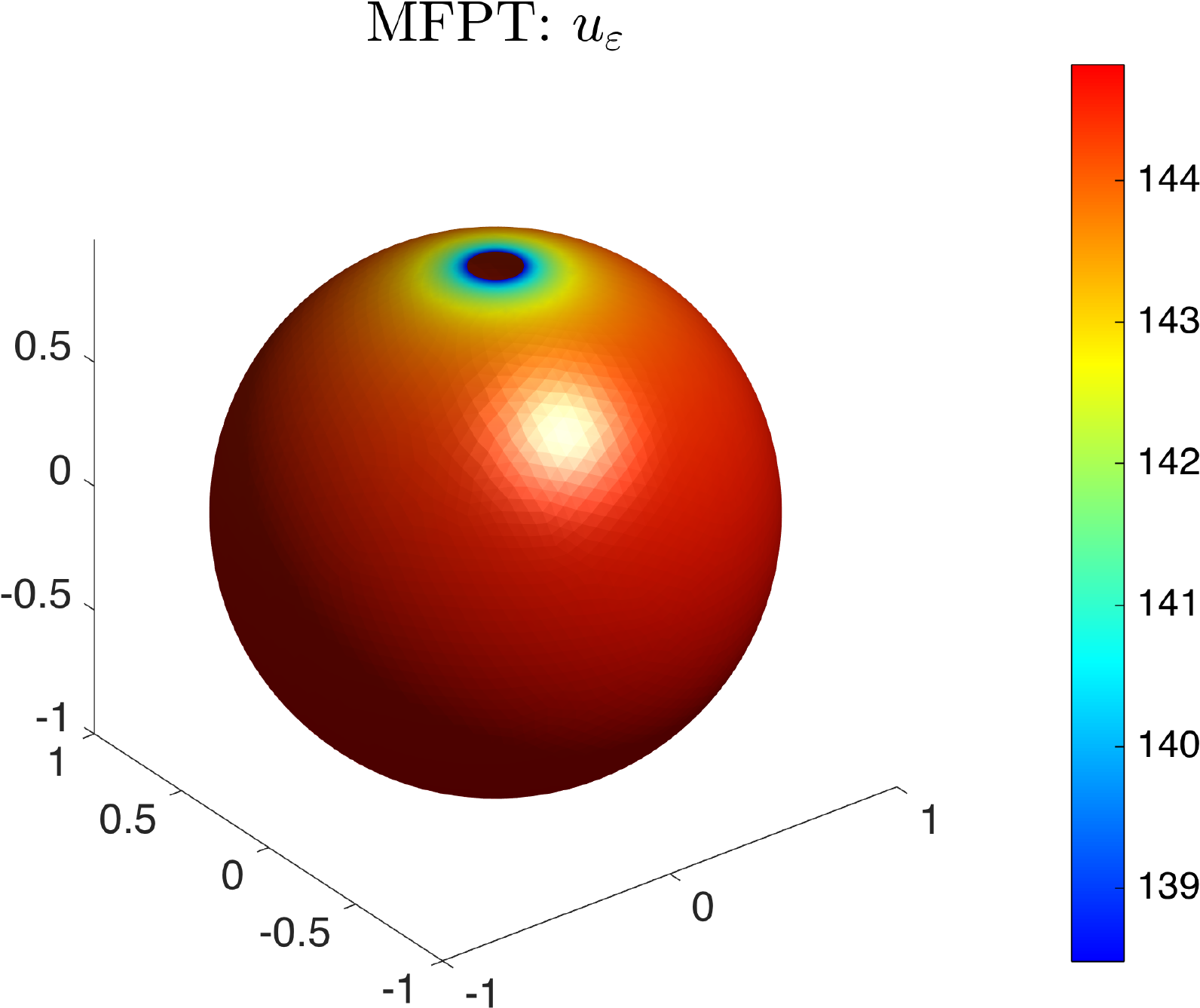}
    \caption{}
    \end{subfigure}
    \begin{subfigure}[b]{0.3\textwidth}
     \includegraphics[width=\textwidth]{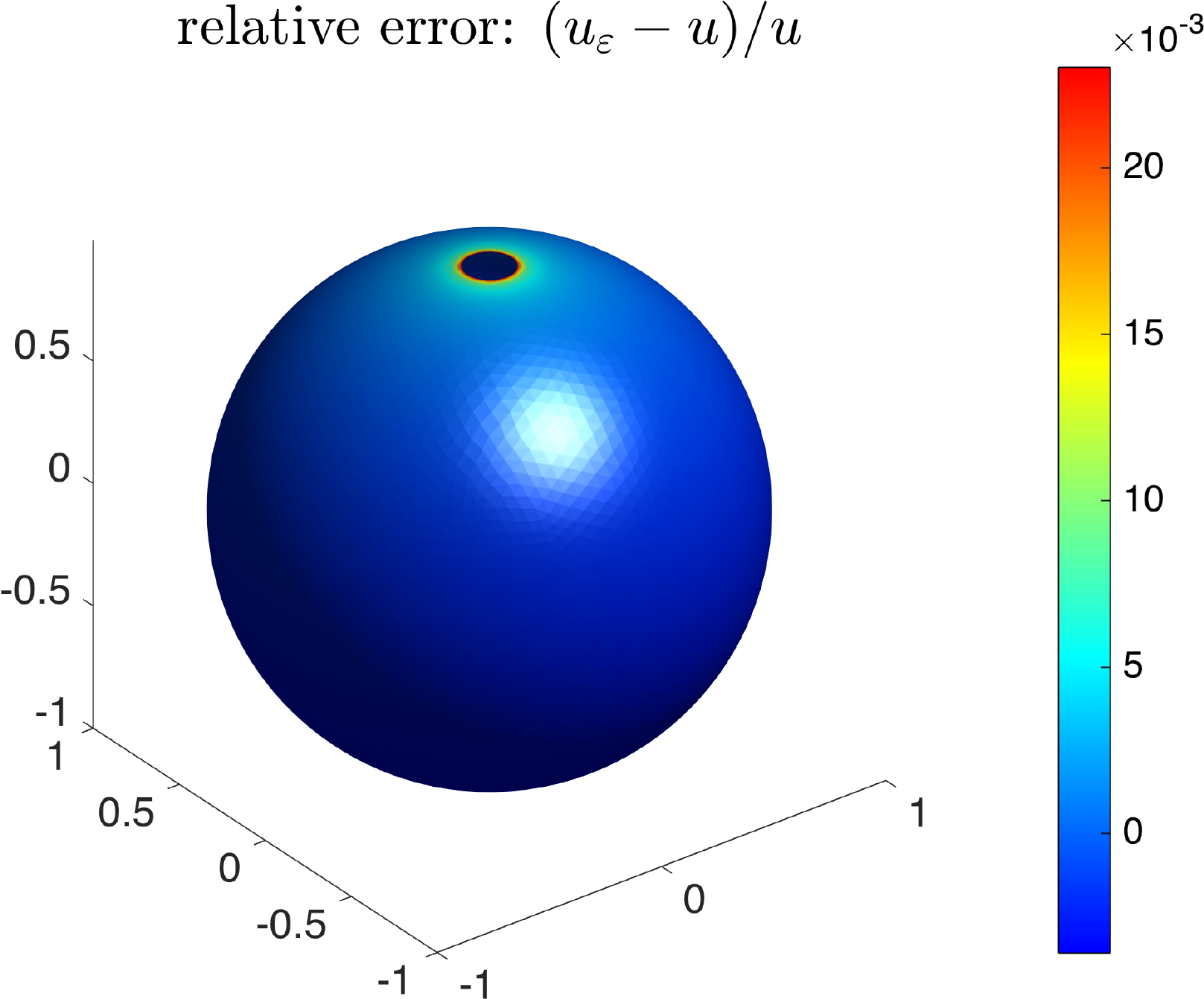}
     \caption{}
     \end{subfigure}
    \caption{(a) numerical solution $u$ on the surface of the spine; (b) asymptotic solution $u_\varepsilon$ on the surface of the spine head; (c) relative error between $u_{\varepsilon}$ and $u$ on the surface of the spine head.}
    \label{com01}
    \end{figure}

For the first experiment, we select the spine head as an unit ball and set $\varepsilon=0.1, L=1.0$. In Figure~\ref{com01}(a) we plot the numerical solution of $u$ on $\partial\Omega_h$, i.e. the whole surface of the spine. Observe that the MFPT $u$ is relatively large in the spine head and decreases monotonically to zero towards the end of the spine neck. This is consistent with our intuition about the underlying physical process. In Figure \ref{com01}(b) we plot the asymptotic solution $u_\varepsilon$ according to \eqref{eq:5} on $\partial\Omega_h$, i.e. the surface of the spine head. Note that $u$ is relatively constant but is smaller for points closer to the connection part, which is also consistent with the physical intuition. In Figure \ref{com01}(c) we plot the relative error between $u$ and $u_\varepsilon$, computed as $(u_\varepsilon-u)/u$, on the surface of the spine head. The maximal relative error is seen to be about $0.2\%$.

For the next experiment, we fix the neck length $L=1.0$ and let the neck radius $\varepsilon$ decreases from $0.1$ to $0.01$ in a step size of $0.01$. For each value of $\varepsilon$, we compare the value of the numerical solution $u$ to the original problem \eqref{main_equation}, the numerical solution of the Robin-Neumann model \eqref{NR_eq}, which is denoted by $u_r$, and the asymptotic solution $u_\varepsilon$ given in \eqref{eq:5}. In Table \ref{tab:epsilon} we list the value of $u,u_r,u_\varepsilon$, as well as the relative error $(u_\varepsilon-u)/u$, at the center of the spine head. Clearly we have a good match of the solutions and small relative error for all values of $\varepsilon$.
\begin{table}
\centering
\begin{tabular}{rrrrr}
  \toprule
  \multicolumn{1}{c}{$\varepsilon$}
  &
    \multicolumn{1}{c}{$u$}
  &
    \multicolumn{1}{c}{$u_r$}
  &
    \multicolumn{1}{c}{$u_{\varepsilon}$}
  &
    \multicolumn{1}{c}{$(u_{\varepsilon}-u)/u$}
  \tabularnewline
  \midrule
  0.10 &   145.01 &   145.37 &   144.48 &  -0.0036 \tabularnewline
  0.09 &   177.53 &   177.45 &   177.01 &  -0.0029 \tabularnewline
  0.08 &   222.81 &   223.28 &   222.31 &  -0.0022 \tabularnewline
  0.07 &   288.57 &   289.11 &   288.11 &  -0.0016 \tabularnewline
  0.06 &   389.47 &   390.12 &   389.07 &  -0.0010 \tabularnewline
  0.05 &   556.12 &   556.91 &   555.80 &  -0.0006 \tabularnewline
  0.04 &   861.64 &   862.63 &   861.46 &  -0.0002 \tabularnewline
  0.03 &  1518.90 &  1520.30 &  1519.04 &   0.0001 \tabularnewline
  0.02 &  3389.10 &  3391.10 &  3389.75 &   0.0002 \tabularnewline
  0.01 & 13443.00 & 13456.00 & 13446.34 &   0.0002 \tabularnewline
  \bottomrule
\end{tabular}
\caption{Solutions and relative error for different values of $\varepsilon$.}
\label{tab:epsilon}
\end{table}

The asymptotic solution $u_\varepsilon$ in \eqref{eq:5} at a fixed $x$ can be considered as a quadratic polynomial of $1/\varepsilon$ with leading coefficients $|\Omega_h|L/\pi=4/3$ and $|\Omega_h|M/\pi^2 \approx 1.13$. We now confirm these coefficients by fitting the value of $u$ by a quadratic polynomial of $1/\varepsilon$ in Table \ref{tab:epsilon}. The result is plotted in Figure \ref{fig:fit}. Clearly the coefficients of the fitting polynomial matches well with those of the asymptotic solution.
\begin{figure}[t]
  \centering
  \includegraphics[width=0.6\textwidth]{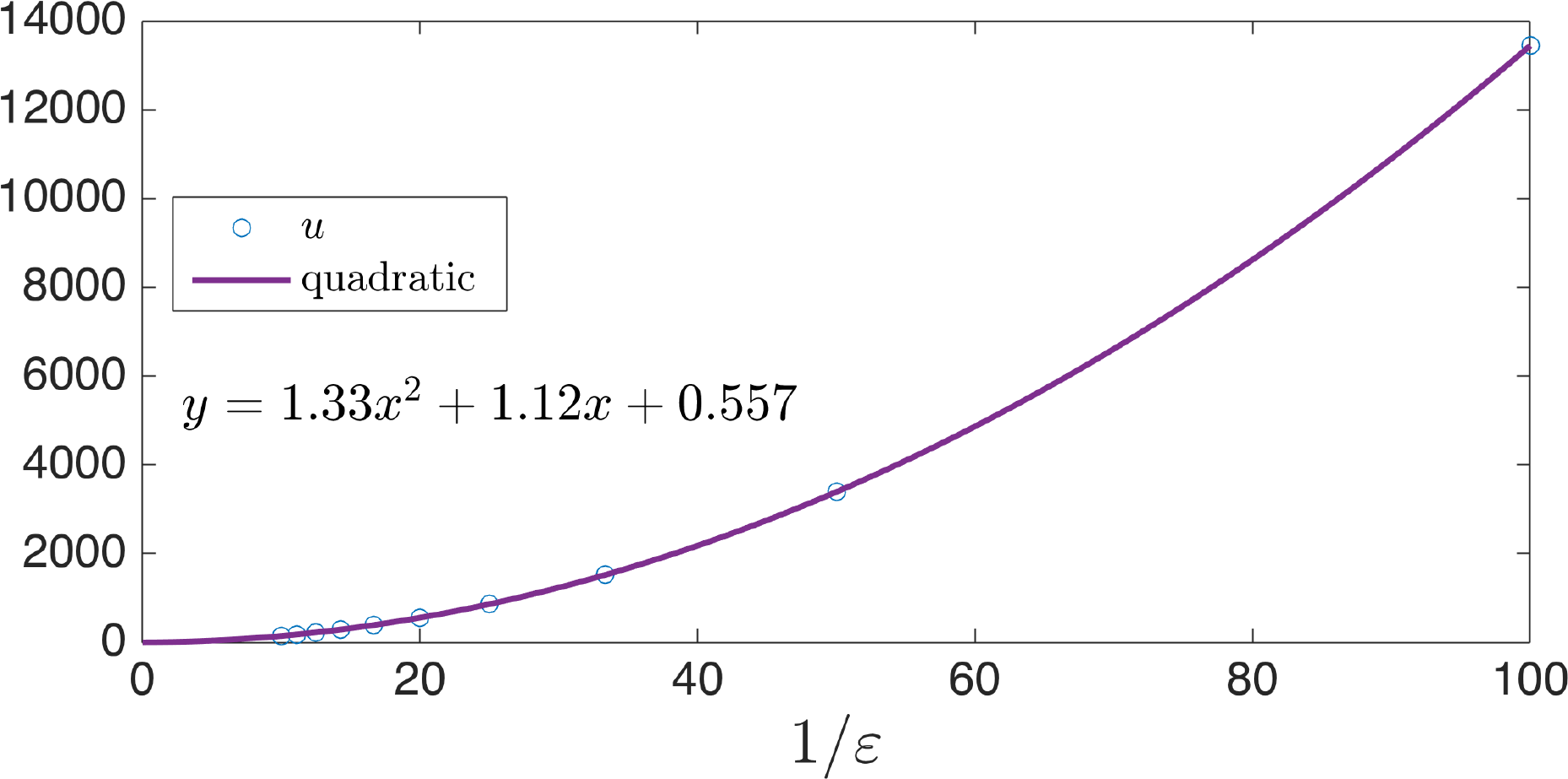}
  \caption{Quadratic fitting of the data in the first two columns of Table \ref{tab:epsilon}.}
  \label{fig:fit}
\end{figure}

Another important parameter in the asymptotic solution \eqref{eq:5} is $L$, the length of the spine neck. For the next experiment, we fix the neck radius $\varepsilon=0.05$ and let the neck radius $L$ increases from $1.0$ to $10.0$ in a step size of $1.0$. In Table \ref{tab:L} we list the value of $u,u_r,u_\varepsilon$, as well as the relative error $(u_\varepsilon-u)/u$, at the center of the spine head. Clearly we have a good match of the solutions and small relative error for each value of $L$.
\begin{table}
  \centering
  \begin{tabular}{rrrrr}
    \toprule
    \multicolumn{1}{c}{$L$}
    &
      $u$
    &
      \multicolumn{1}{c}{$u_r$}
    &
      \multicolumn{1}{c}{$u_{\varepsilon}$}
    &
      \multicolumn{1}{c}{$(u_{\varepsilon}-u)/u$}
    \tabularnewline
    \midrule
    1.0 &   555.98 &   556.91 &   555.80 &  -0.0003 \tabularnewline
    2.0 &  1090.80 &  1091.80 &  1090.63 &  -0.0002 \tabularnewline
    3.0 &  1626.60 &  1627.60 &  1626.47 &  -0.0001 \tabularnewline
    4.0 &  2163.40 &  2164.50 &  2163.30 &  -0.0000 \tabularnewline
    5.0 &  2700.90 &  2702.30 &  2701.13 &   0.0001 \tabularnewline
    6.0 &  3239.90 &  3241.20 &  3239.97 &   0.0000 \tabularnewline
    7.0 &  3779.70 &  3781.10 &  3779.80 &   0.0000 \tabularnewline
    8.0 &  4320.20 &  4321.90 &  4320.63 &   0.0001 \tabularnewline
    9.0 &  4865.10 &  4863.80 &  4862.47 &  -0.0005 \tabularnewline
    10.0 &  5408.30 &  5406.60 &  5405.30 &  -0.0006 \tabularnewline
    \bottomrule
  \end{tabular}
  \caption{Solutions and relative error for different values of $L$.}
  \label{tab:L}
\end{table}

Finally we conduct numerical experiments on three different shapes of the spine head, which may correspond to different types of spine. The numerical solution $u$, the asymptotic solution $u_\varepsilon$ and the relative error between them are shown in  Figure \ref{fig:three}. The first row shows the results when the spine head is a relatively flat three dimensional domain. The second one is a relatively thin domain and the third one is a non-convex domain. From the relative error $(u_\varepsilon-u)/u$, which is described in the third column, we can see that the relative error is small enough to show that asymptotic formula $u_\varepsilon$ is a good approximation to the MFPT.
\begin{figure}
\centering
\begin{subfigure}[b]{0.31\textwidth}
    \includegraphics[width=\textwidth]{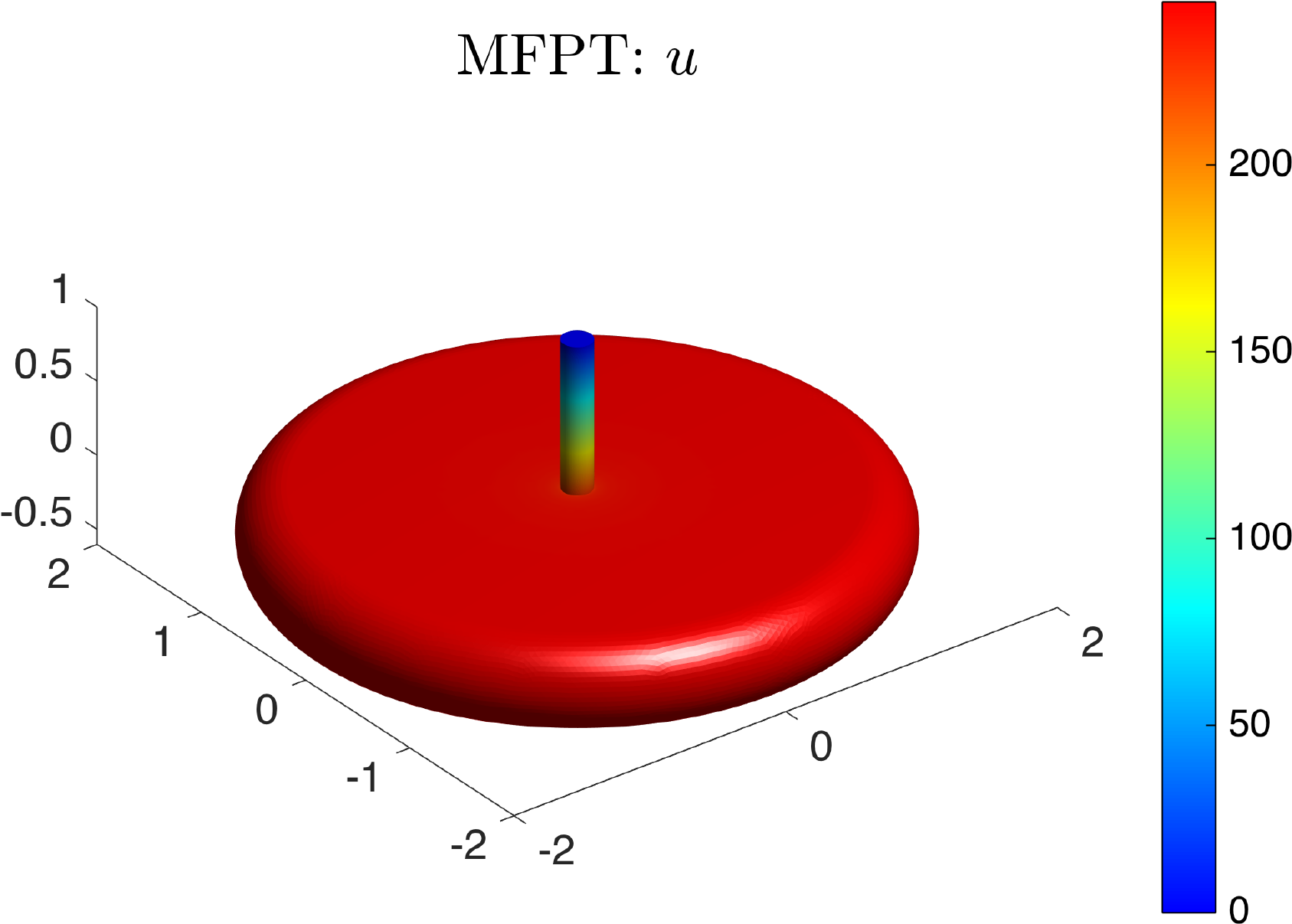}
    \caption{}
    \end{subfigure}
    \begin{subfigure}[b]{0.3\textwidth}
    \includegraphics[width=\textwidth]{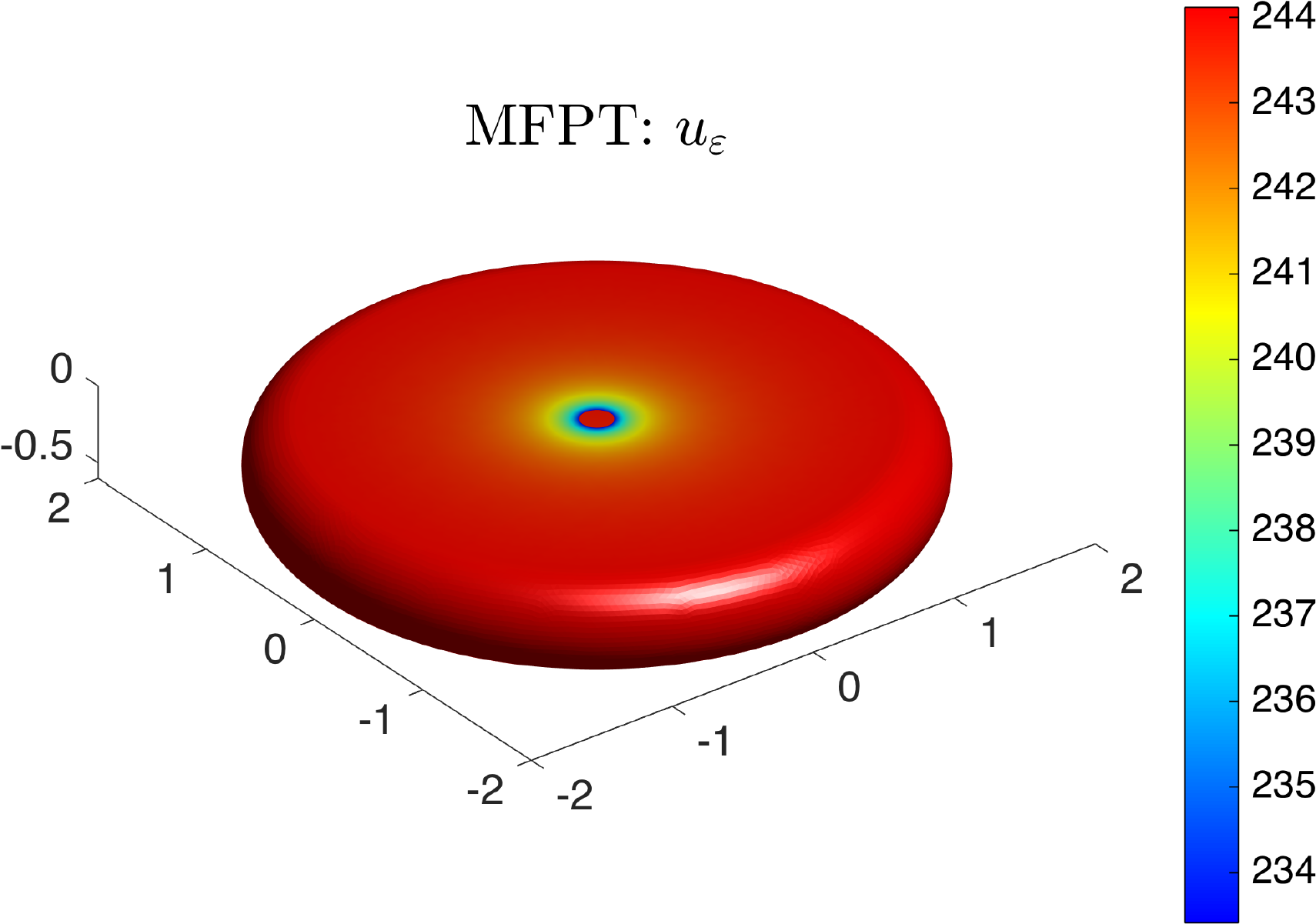}
    \caption{}
    \end{subfigure}
     \begin{subfigure}[b]{0.3\textwidth}
    \includegraphics[width=\textwidth]{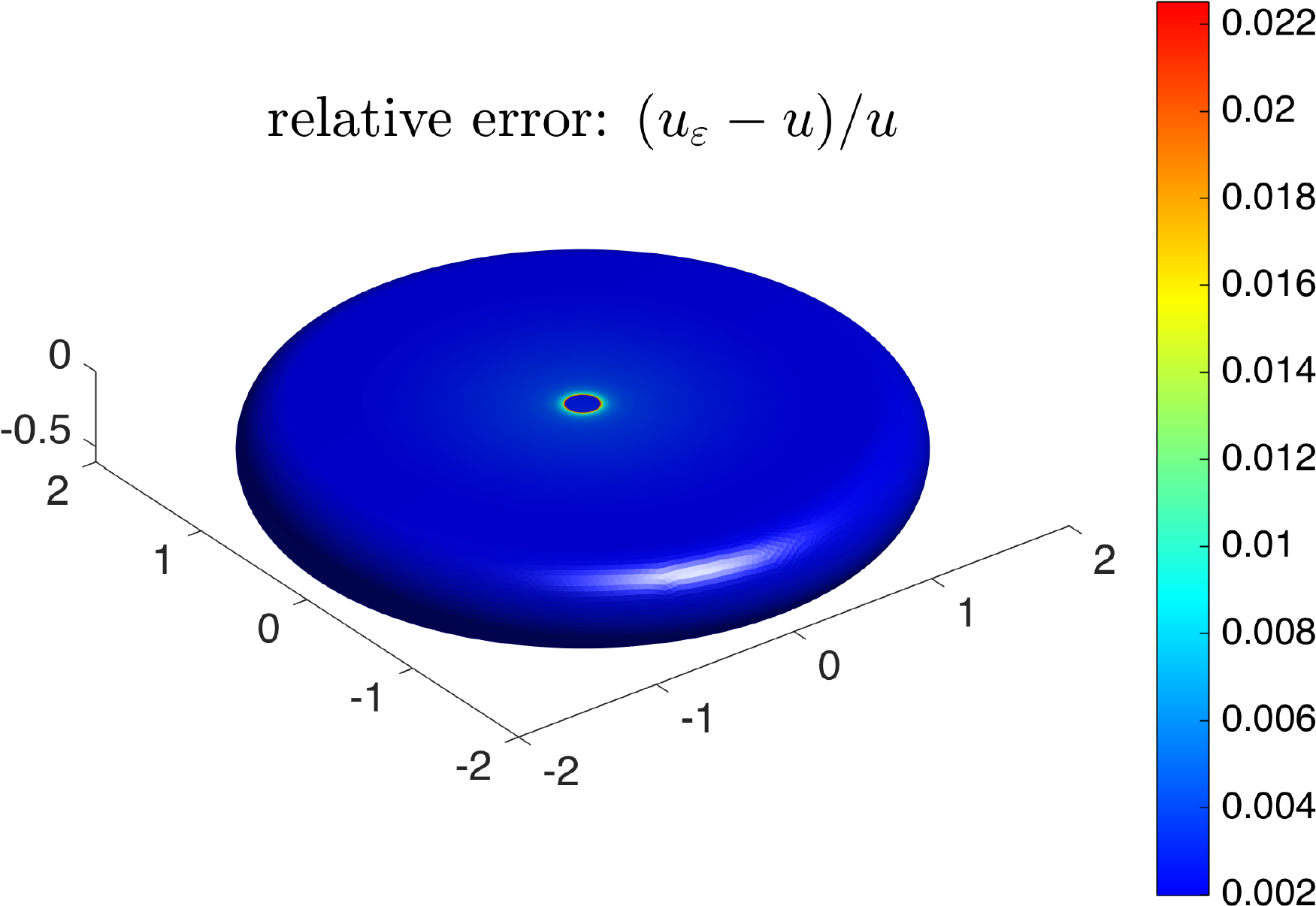}
    \caption{}
    \end{subfigure}
    \begin{subfigure}[b]{0.33\textwidth}
     \includegraphics[width=\textwidth]{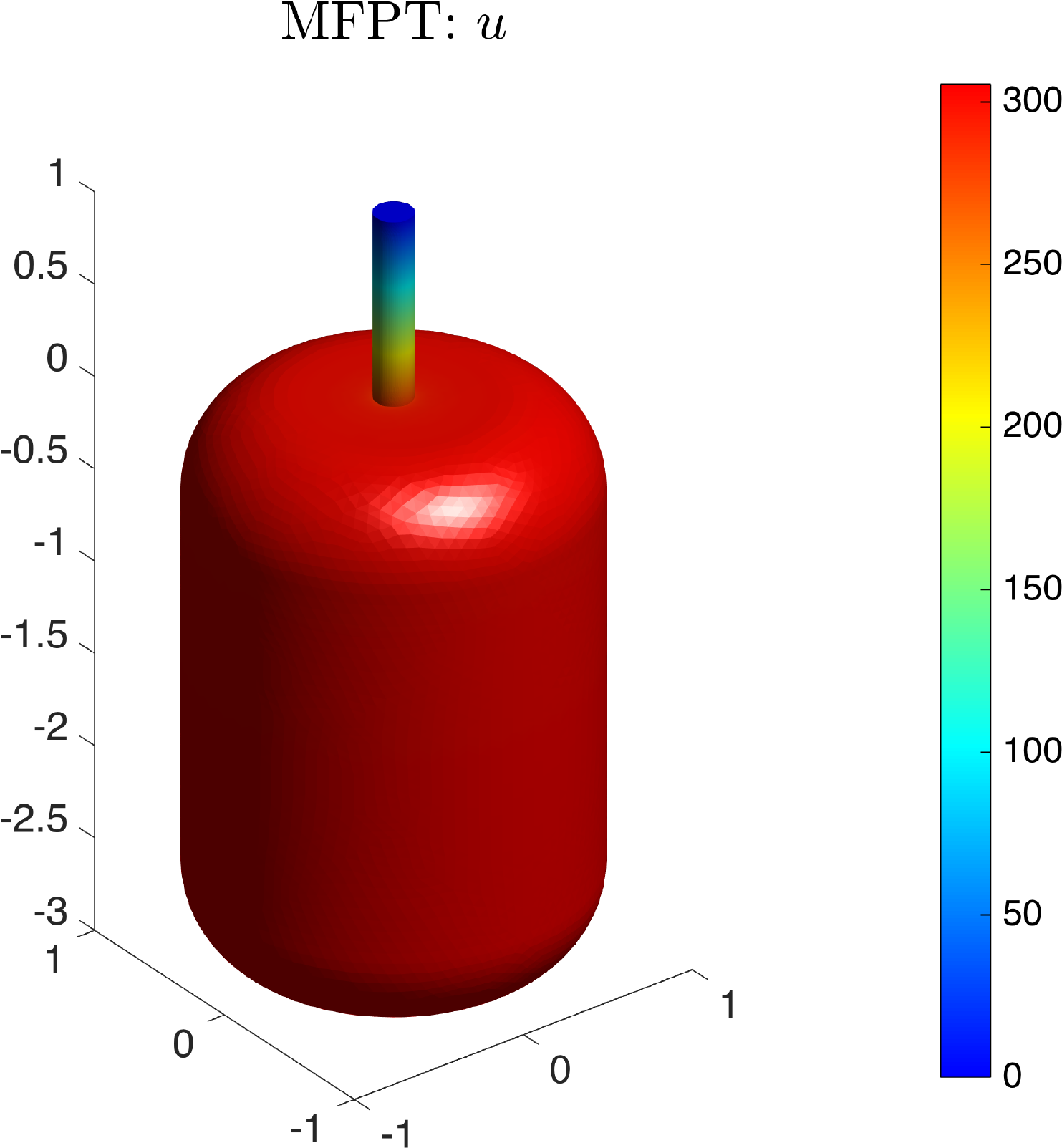}
     \caption{}
     \end{subfigure}
     \begin{subfigure}[b]{0.3\textwidth}
      \includegraphics[width=\textwidth]{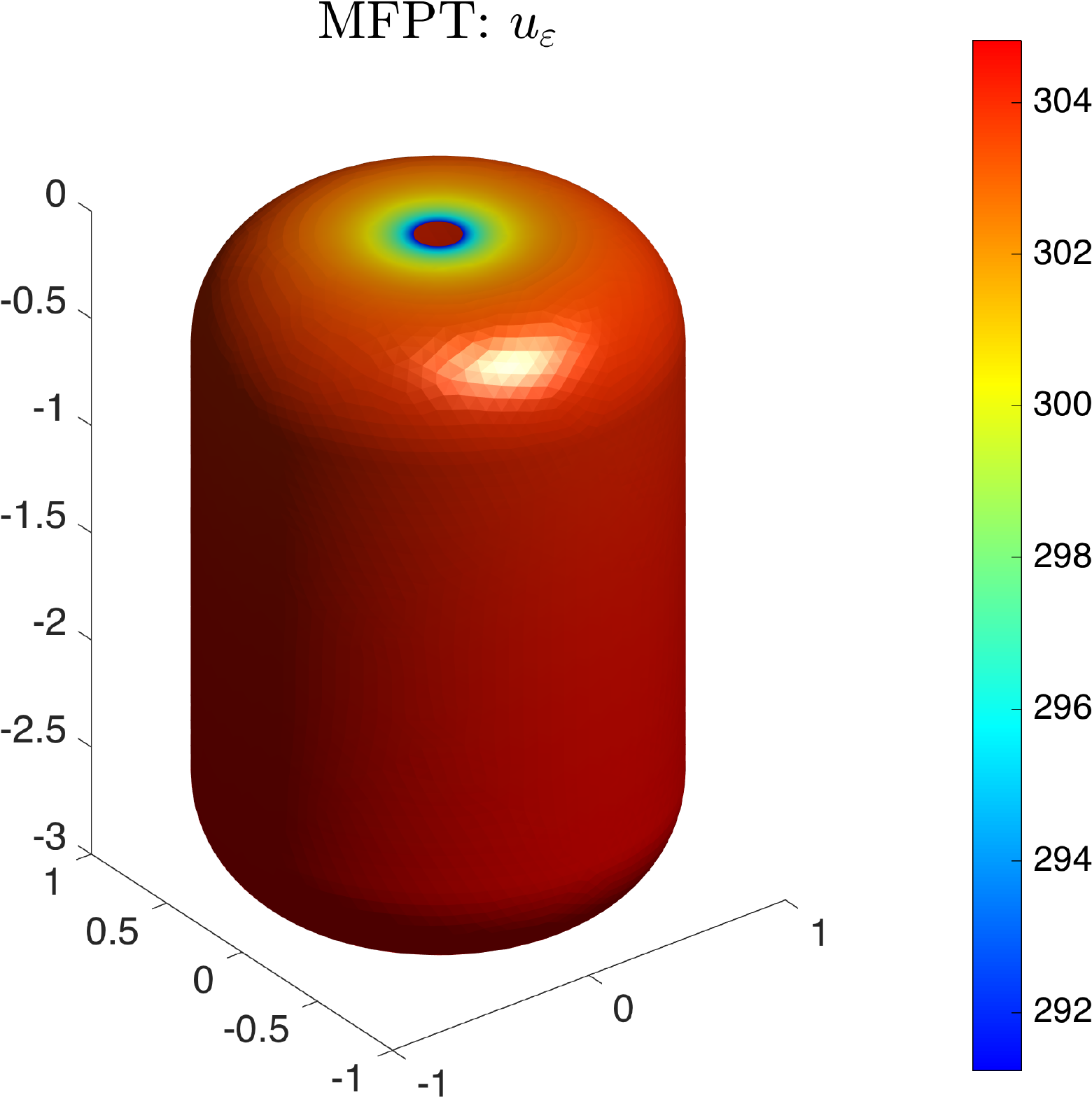}
      \caption{}
      \end{subfigure}
       \begin{subfigure}[b]{0.3\textwidth}
    \includegraphics[width=\textwidth]{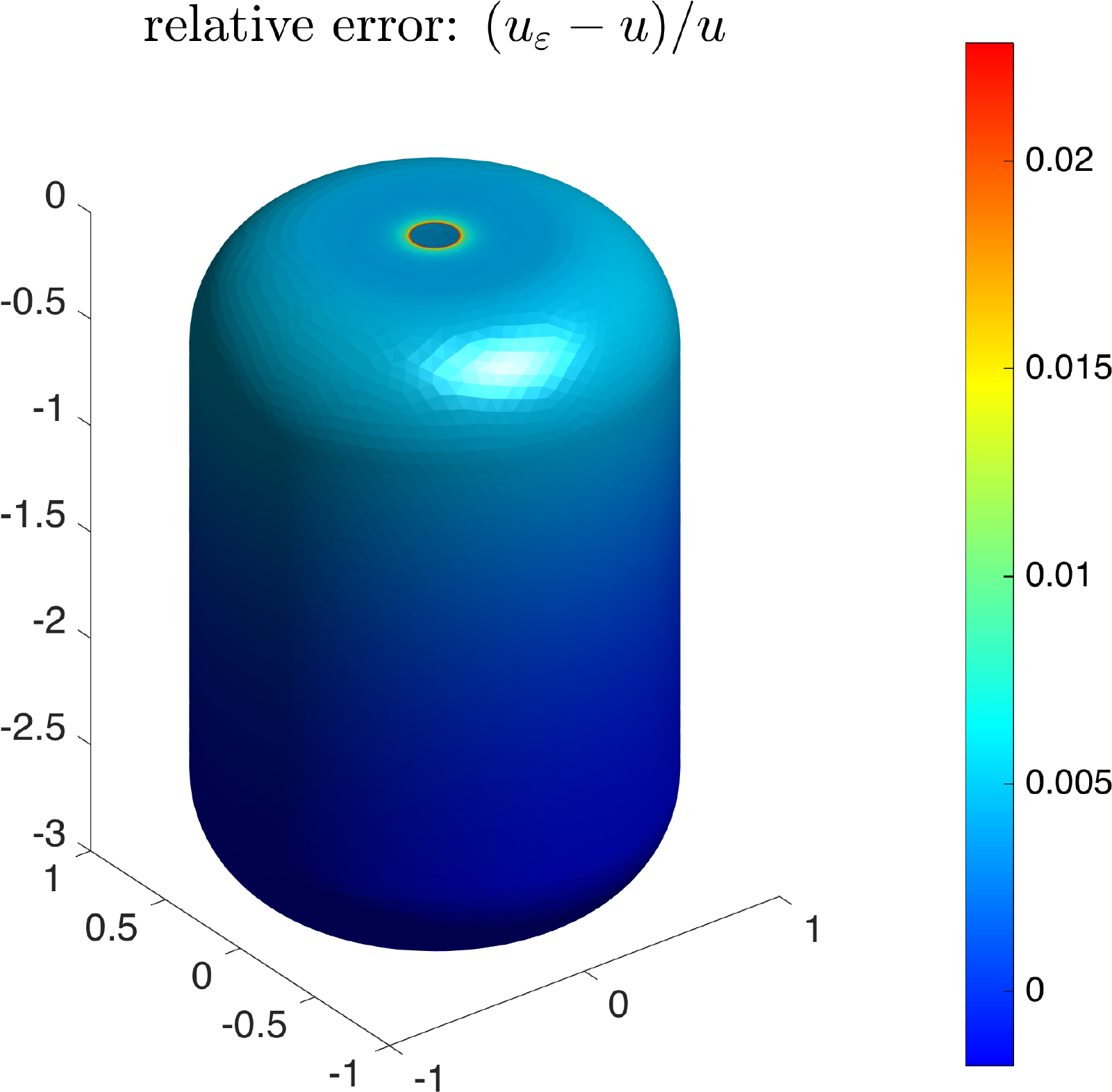}
    \caption{}
    \end{subfigure}
     \begin{subfigure}[b]{0.31\textwidth}
    \includegraphics[width=\textwidth]{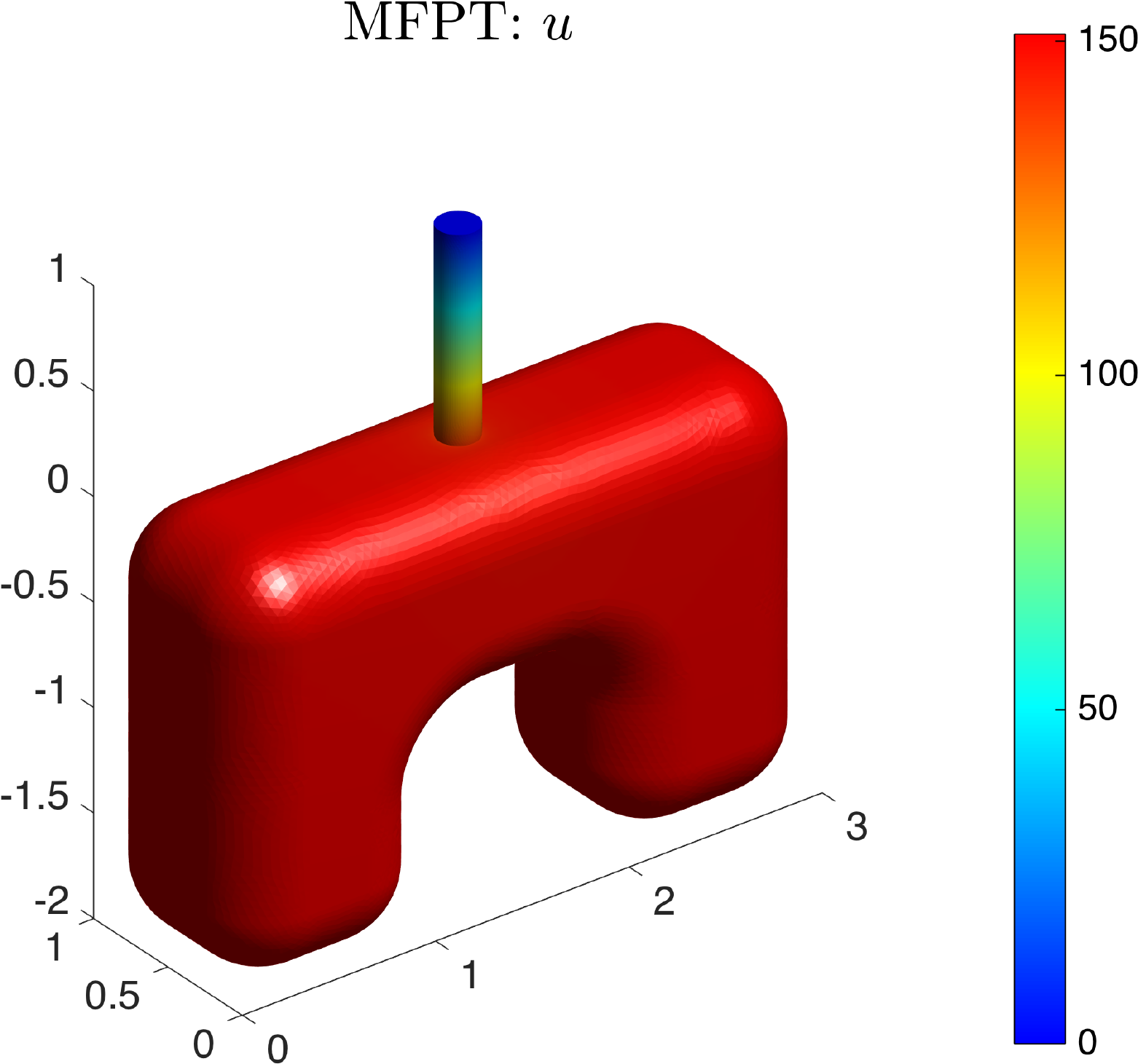}
    \caption{}
    \end{subfigure}
     \begin{subfigure}[b]{0.3\textwidth}
    \includegraphics[width=\textwidth]{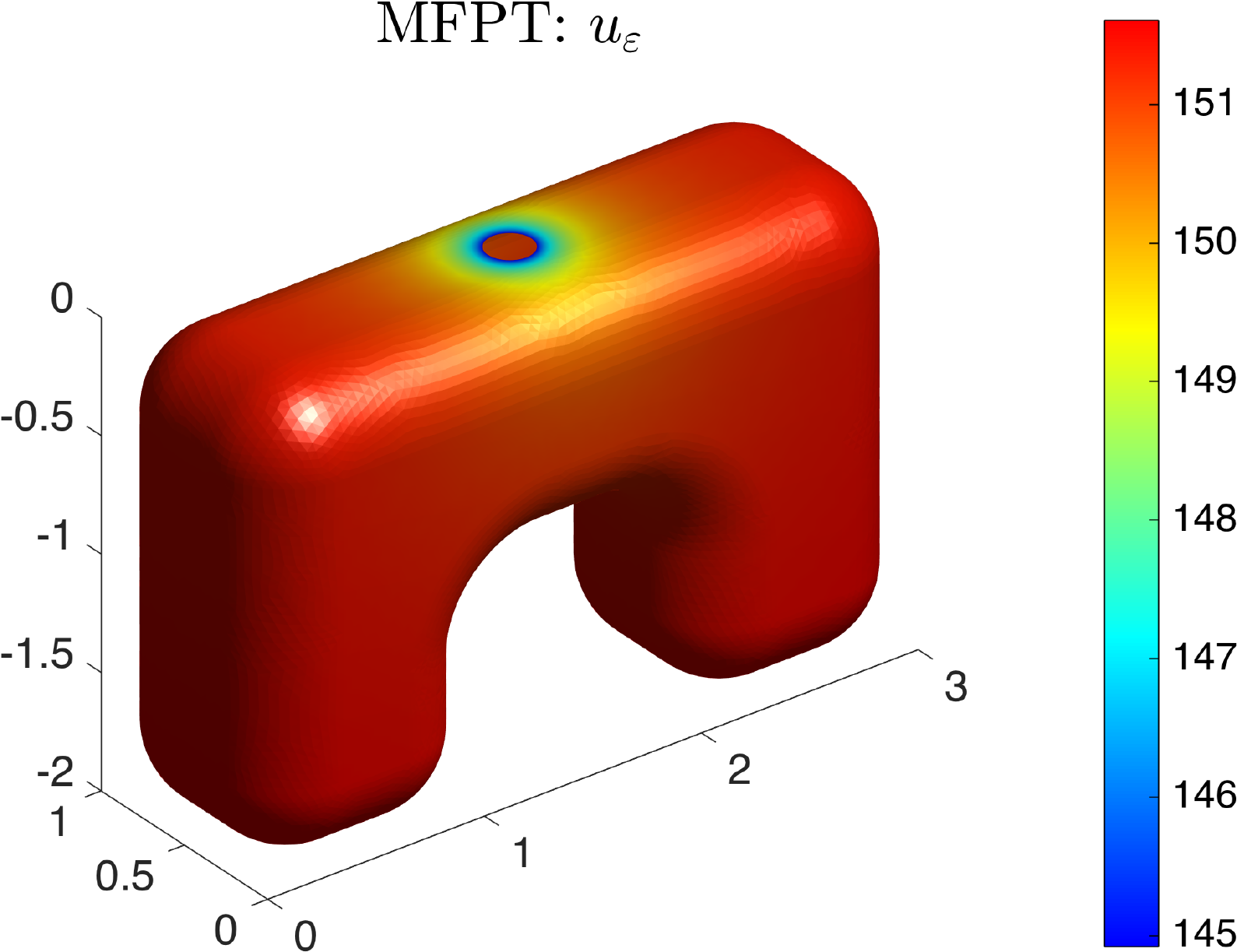}
    \caption{}
    \end{subfigure}
     \begin{subfigure}[b]{0.3\textwidth}
    \includegraphics[width=\textwidth]{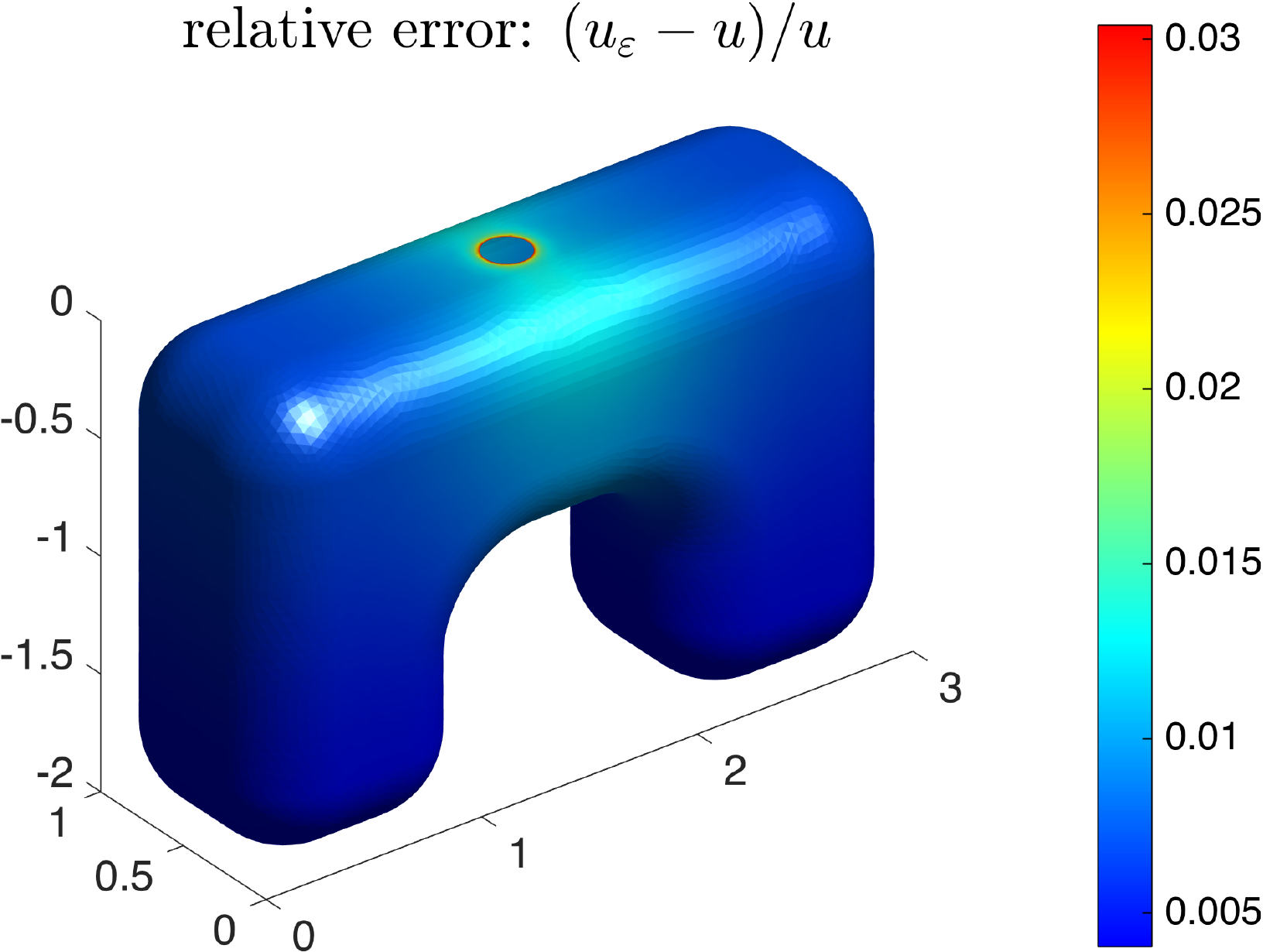}
    \caption{}
    \end{subfigure}
    \caption{The numerical solution $u$, the asymptotic solution $u_\varepsilon$ and the relative error $(u_\varepsilon-u)/u$ for three different types of spine.}
    \label{fig:three}
    \end{figure}

\section{Conclusion}
In this study, we used the Robin-Neumann model to compute the NEP in a three-dimensional
domain with a long neck, which is often referred as a dendritic spine domain. This is a following up paper of \cite{Li}, where the Robin-Neumann model is presented to solve the NEP in a two-dimensional analogue of a dendritic
spine domain. In this paper, we derived asymptotic expansion formula for three-dimensional Robin-Neumann model. Our results demonstrate that this asymptotic expansion formula could approximate the MFPT up to at least second leading order using this model, which has not been reported previously.

The work of Hyundae Lee was supported by the National Research Foundation of Korea grant (NRF-2015R1D1A1A01059357). The work of Yuliang Wang was supported by the Hong Kong RGC grant (No.~12328516) and the NSF of China (No.~11601459).

\section{Appendix}
\subsection{Appendix A.}

\begin{lem}(Single layer potential on an surface) Let $\Gamma_1 \subset \mathbb{R}^2$ be a bounded domain with $ \sup_{x,z \in \Gamma_1} |x-z| \leq 2$. The integral operator $L$: $L^\infty(\Gamma_{1})\mapsto L^\infty(\Gamma_1)$ defined by
$$L[\phi](x)=\int_{\Gamma_1}\frac{1}{|x-z|}\phi(z)dz$$
is bounded, i.e.
$$\|S[\phi]\|_{L^\infty(\Gamma_{1})}\leq C\|\phi\|_{L^{\infty}(\Gamma_{1})},$$
where $C$ is a constant independent of $x$.
\end{lem}
\begin{proof}
Let $B_{x,2} \subset \mathbb{R}^2$ denote the disk centered at $x$ with radius $2$. For given $x \in \Gamma_1$ we have
  \begin{equation*}
    \int_{\Gamma_1} \dfrac{1}{|x-z|} \, {\rm d}z
    \leq \int_{B_{x,2}} \dfrac{1}{|x-z|} \, {\rm d}z
    = \int_0^{2\pi} \int_0^2 \frac{1}{\rho} \rho \, {\rm d}\rho {\rm d}\theta = 4 \pi.
  \end{equation*}
  Hence
    \begin{equation*}
      \left| \int_{\Gamma_1} \dfrac{1}{|x-z|} \phi(z) \, {\rm d}z \right| \leq \|\phi\|_{L^\infty(\Gamma_1)} \int_{\Gamma_1} \dfrac{1}{|x-z|} \, {\rm d}z \leq 4 \pi  \|\phi\|_{L^\infty(\Gamma_1)}, \quad x \in \Gamma_1.
    \end{equation*}
    Therefore $L$ is bounded.

\end{proof}
\subsection{Appendix B.}
The value of
$$\int_{\Gamma_1}\int_{\Gamma_1}\frac{1}{|x-y|}dxdy$$
is $\frac{16}{3}\pi$, where $\Gamma_1$ is unit disk.

\begin{proof}
First, fix point $y$, draw a circle centered at $y$ with radius $r$, where the distance between $x$ and $y$ is $r$. Let $s$ be the distance between $0$ and $y$.

Then we have
\beq
\int_{\Gamma_1}\int_{\Gamma_1}\frac{1}{|x-y|}dxdy=\int_0^1 2\pi s\int_0^{1-s}2\pi drds+\int_0^12\pi s\int^{1+s}_{1-s}s\arccos\frac{s^2+r^2-1}{2sr} drds.
\label{app}
\eeq

When $r\leq 1-s$, by simple calculation, we can calculate the first term of (\ref{app}) as
$$\int_0^12\pi s\int_0^{1-s}2\pi drds=\frac{2}{3}\pi^2.$$

Changing the order of $ds$ and $dr$, the second term of (\ref{app}) becomes
\beq
4\pi\int^{1}_{0}\int^{1}_{1-r}s\arccos\frac{s^2+r^2-1}{2sr}
dsdr+4\pi\int^{2}_{1}\int^{1}_{r-1}s \arccos\frac{s^2+r^2-1}{2sr} dsdr.
\label{app2}
\eeq

By calculations, we get the value of the first term of (\ref{app2}) as $\frac{4}{3}+\frac{\pi}{6}-\frac{3}{4}\sqrt{3}$, and second term $-\frac{1\pi}{3}+\frac {3}{4}\sqrt{3}$.

Finally, we get
$$\int_{\Gamma_1}\int_{\Gamma_1}\frac{1}{|x-y|}dxdy=\frac{16}{3}\pi.$$
\end{proof}

    \end{document}